\documentclass[onecolumn, 11pt]{IEEEtran}

\usepackage{mathbf-abbrevs}
\usepackage{amsmath}

\newcommand{\reals}{{\mathbb R}}
\newcommand{\ints}{{\mathbb Z}}

\newcommand{\vor}{\operatorname{Vor}}

\newcommand{\term}{\emph}

\newcommand{\prob}{\operatorname{Pr}}

\newcommand{\abs}[1]{{\left| #1 \right|}}

\newcommand{\dotprod}[2]{ #1^\prime #2}

\usepackage[square,comma,numbers,sort&compress]{natbib}

\usepackage{units}
		
\usepackage{booktabs}
		
\usepackage{ifpdf}
\ifpdf
  \usepackage[pdftex]{graphicx}
  \usepackage[naturalnames]{hyperref}
\else
	\usepackage{graphicx}
\fi

\usepackage{amsmath,amsfonts,amssymb, amsthm, bm} 

\usepackage[vlined, linesnumbered]{algorithm2e}
	
\usepackage{mathrsfs}

\newtheorem{theorem}{Theorem}

\newtheorem{example}{Example}

\newtheorem{lemma}{Lemma}

\title{On the error performance of the $A_n$ lattices}
\author{Robby McKilliam, Ramanan Subramanian, Emanuele Viterbo, I. Vaughan L. Clarkson}

\begin{document}

\newcommand{\calR}{\mathcal{R}}
\newcommand{\hist}{\operatorname{hist}}

\maketitle

\begin{abstract}
We consider the root lattice $A_n$ and derive explicit formulae for
the moments of its Voronoi cell.  We then show that these formulae
enable accurate prediction of the error probability of
lattice codes constructed from $A_n$.
\end{abstract}


\section{Introduction}\label{sec:introduction}

The \emph{root lattices} $A_n$, $D_n$, $E_6$, $E_7$ and $E_8$ have attracted particular attention as structured codes for the AWGN channel~\cite{SPLAG}. The highly symmetric structure of these lattices provides the grounds for extremely efficient decoding algorithms~\cite{Conway1982FastQuantDec,McKilliam2009CoxeterLattices,McKilliam2008}.  In this paper we consider codes constructed from the root lattice $A_n$ and derive formulae for accurately predicting the performance of these codes.  This is achieved by deriving formulae for the \emph{moments} of the \emph{Voronoi cell} of $A_n$.  Conway and Sloane suggested this approach to compute the quantizing constants (second order moments) of the root lattices,~\cite{Conway1982VoronoiRegions}.  In this paper we extend their technique to compute the moments of any order for $A_n$.

In two dimensions $A_2$ is the hexagonal lattice and in three dimensions $A_3$ is the face-centered cubic lattice.  These are the densest known sphere packings in dimensions two and three and our results automatically include low dimensional codes constructed using these packings.  In general, the lattice $A_n$ does not produce asymptotically good codes in large dimensions, but does offer a coding gain in small dimensions.  For these cases, we provide an error probability expression that can be computed to any degree of accuracy at any finite signal-to-noise ratio.

This paper is organised as follows.  In Section~\ref{sec:latt-latt-codes} we give a brief overview of lattices and codes constructed from them, i.e., \emph{lattice codes}.  We describe lattice decoding and show how the probability of coding error can be expressed in terms of the moments of the \emph{Voronoi cell} of the lattice used.  Section~\ref{sec:main-result} states the main result (without proof), which gives recursive formula to compute the moments of the Voronoi cell of $A_n$.  Section~\ref{sec:lattice-a_n} describes the lattice $A_n$ and some of its properties.  An important property for our purposes is that the Voronoi cell of $A_n$ is precisely the projection of a $n+1$-dimensional hypercube orthogonal to one of its vertices~\cite{McKilliam2009CoxeterLattices,McKilliam2010thesis}.
In Section~\ref{sec:integr-funct-over} we use this property to show how integrals over the Voronoi cell of $A_n$ can be expressed as integrals over the $n$-dimensional hypercube. These integrals are solvable and we use them to obtain the moments of the Voronoi cell in Section~\ref{sec:powers-eucl-norm}.  In Section~\ref{sec:results-simulations} we plot the probability of error versus signal to noise ratio for codes constructed from the lattices $A_1 \simeq \ints$, $A_2$, $A_3 \simeq D_3$ and $A_4$.  We also plot the results of Monte-Carlo simulations that support our analytical results.

\section{Lattices, lattice codes, and lattice decoding} \label{sec:latt-latt-codes}

A \term{lattice}\index{lattice}, $\Lambda$, is a discrete subset of points in $\reals^m$ such that
\[
   \Lambda = \{\xbf = \Bbf\ubf \mid \ubf \in \ints^n \}
 \]
 where $\Bbf \in \reals^{m \times n}$ is an $m \times n$ matrix of rank $n$, called the \term{generator matrix} or \term{basis matrix} or simply generator or basis. In particular, the set of $n$-tuples of integers $\ints^n$ is a lattice (with the identity matrix as the generator) and we call this the~\emph{integer lattice}. A lattice $\Lambda$ associated with a rank-$n$ generator matrix $\mathbf{B}$ is said to be $n$-dimensional. If the generator is square, i.e. $m = n$, then the lattice points span $\reals^n$ and we say that the lattice is \term{full rank}. If $\Bbf$ has more rows than columns, i.e. $m > n$, then the lattice points lie in a $n$-dimensional subspace of $\reals^m$. For any lattice $\Lambda$ with an $m \times n$ generator matrix, we define $\mathcal{S}_{\Lambda}$ to be the hyperplane spanned by the columns of the generator matrix. It is easy to see that $\mathcal{S}_\Lambda$ is then the $n$-dimensional subspace containing the lattice $\Lambda$.

The (open) \term{Voronoi cell}, denoted $\vor(\Lambda)$, of a lattice $\Lambda$ is the subset of $\mathcal{S}_{\Lambda}$ containing all points nearer (in Euclidean distance) to the lattice point at the origin than any other lattice point. It can be shown that the Voronoi cell is an $n$-dimensional convex polytope that is symmetric about the origin.  It is convenient to modify this definition of the Voronoi cell slightly so that the union of translated Voronoi cells $\cup_{\xbf \in \Lambda}\vor(\Lambda) + \xbf$ is equal to  $\mathcal{S}_{\Lambda}$.  That is, the Voronoi cell tessellates space when translated by points in $\Lambda$.  To ensure this we require that if a face of $\vor(\Lambda)$ is open, then its opposing face is closed. Specifically, if $\xbf \in \vor(\Lambda)$ is on the boundary of $\vor(\Lambda)$ then $-\xbf \notin \vor(\Lambda)$.  We wont specifically define which opposing face is open and which is closed as the results that follow hold for any choice of open and closed opposing faces.

The Voronoi cell encodes many interesting lattice properties such as the packing radius, covering radius, kissing number, minimal vectors, center density, thickness, and the normalized second moment (or quantizing constant)~\cite{Viterbo_diamond_cutting_1996, SPLAG}. The error probability of a lattice code can also be evaluated from the Voronoi cell as we will see.  There exist algorithms to completely enumerate the Voronoi cell of an arbitrary lattice~\cite{Viterbo_diamond_cutting_1996,Sikiric_complex_algs_vor_cells_2009,Sikiric_vor_reduction_covering_2008,Valentin2003_coverings_tilings_low_dimension}.  In general these algorithms are only computationally feasible when the dimension is small (approximately $n \leq 9$).  Even with a complete description of the Voronoi cell it is not necessarily easy to compute the probability of coding error.

The Voronoi cell is linked with the problem of \emph{lattice decoding}.  Given some point $\ybf \in \reals^n$ a \emph{lattice decoder} (or \emph{nearest lattice point algorithm}) returns the lattice point in $\Lambda$ that is nearest to $\ybf$~\cite{Agrell2002}.  Equivalently it returns the lattice point $\xbf$ such that the translated Voronoi cell $\vor(\Lambda) + \xbf$ contains $\ybf$.  Computationally lattice decoding is known to be NP-hard under certain conditions when the lattice itself, or rather a basis thereof, is considered as an additional input parameter~\cite{micciancio_hardness_2001, Jalden2005_sphere_decoding_complexity}. Nevertheless, algorithms exist that can compute the nearest lattice point in reasonable time if the dimension is small (approximately $n \leq 60$). One such algorithm is the \term{sphere decoder}~\cite{Viterbo_sphere_decoder_1999,Pohst_sphere_decoder_1981,Agrell2002}.
A good overview of these techniques is given by Agrell~et.~al.~\cite{Agrell2002}. Fast nearest point algorithms are known for specific lattices~\cite{Conway1982FastQuantDec, McKilliam2008,McKilliam2009CoxeterLattices,Vardy1993_leech_lattice_MLD}. For example, the root lattices $D_n$ and $A_n$ and their dual lattices $D_n^*$ and $A_n^*$ can be decoded in linear-time, i.e. in a number of operations of order $O(n)$~\cite{Conway1982FastQuantDec,McKilliam2009CoxeterLattices}.


\newcommand{\calX}{\mathcal X}
Lattices can be used to construct \emph{signal space codes}.  A signal space code $\calX$ of dimension $n$ is a finite set of vectors (or points) in the Euclidean space $\reals^n$.  Each vector in $\calX$ is called a \emph{codeword} and represents a particular signal.  The number of codewords is denoted by $\abs{\calX}$.  If each codeword is transmitted with equal probability then the rate of the code is $R = \frac{1}{n}\log_2\abs{\calX}$ bits per codeword. The average power of the code is given by $P = \frac{1}{\abs{\calX}}\sum_{\xbf \in \calX}\|\xbf\|^2$.  In the AWGN channel the received signal takes the form
\[
\ybf = \xbf + \wbf
\]
where $\ybf \in \reals^n$, $\xbf \in \calX$ and $\wbf$ is a vector of independent and identically distributed Gaussian random variables with variance $\sigma^2$.  If the receiver employs maximum likelihood decoding then the estimator of $\xbf$ given $\ybf$ at the receiver is
\begin{equation}\label{eq:mldecoder}
\hat{\xbf} = \underset{\xbf \in \calX}{\operatorname{argmin}}\;\| \ybf - \xbf \|^2.
\end{equation}
That is, the receiver computes the codeword in $\calX$ nearest in Euclidean distance to the received signal $\ybf$.  Assuming that each codeword is transmitted with equal probability then the probability of correct decoding is
\[
P_C =\frac{1}{\abs{\calX}}\sum_{\xbf \in \calX} \prob( \hat{\xbf} = \xbf ),
\]
and the probability of error is
\[
P_E = 1 - \frac{1}{\abs{\calX}} \sum_{\xbf \in \calX} \prob(
\hat{\xbf} = \xbf ).
\]


A \emph{lattice code} is a signal space code with codewords taken from a finite subset of points of some lattice $\Lambda$ in $\reals^n$. There are infinitely many ways to choose a finite subset from a lattice, but common approaches make use of a bounded subset of $\reals^n$, called a \emph{shaping region} $S \subset \reals^n$.  The codewords are given by those lattice points inside the shaping region, that is,
\[
\calX = S \cap \Lambda.
\]
Common choices of shaping region are $n$-dimensional spheres, spherical shells, hypercubes, or the Voronoi cell of a \emph{sublattice} of $\Lambda$~\cite{Buda1989_some_opt_codes_structure,Erex2004_lattice_decoding,Conway1983VoronoiCodes}.  A consequence of the finiteness of a lattice code is that the maximum likelihood decoder~\eqref{eq:mldecoder} is \emph{not} equivalent to lattice decoding.  Computing a nearest lattice point will not in general return a lattice point from the code (the decoded lattice point might lie outside the shaping region).

Let $\calX$ be a lattice code constructed from a lattice $\Lambda$.  Assuming that each codeword $\xbf \in \calX$ has an equal probability of being transmitted and that the receiver employs lattice decoding, then the average probability of correct decoding in the AWGN channel is
\begin{align}
P_C &\leq P_C(\Lambda) = \lim_{\abs{\calX}\rightarrow \infty}
 \frac{1}{\abs{\calX}} \sum_{\xbf\in \calX} \prob( \xbf + \wbf \in
\vor(\Lambda) + \xbf )  \nonumber \\
&= \prob(\xbf \in \vor(\Lambda) ) \nonumber \\
&=   \frac{1}{(\sqrt{2\pi}\sigma)^n} \int_{\vor(\Lambda)}
e^{-\|\xbf\|^2 / 2\sigma^2 } d\xbf. \label{eq:PEgaussian}
\end{align}
The probability of error is $P_E \leq 1 - P_C(\Lambda)$, where the upper bound is asymptotically tight for large constellations (the proportion of the codewords near the boundary of the shaping region becomes small).  By expanding $e^x = 1  + x + \frac{x^2}{2} + \dots$ according to its Maclaurin series we obtain
\begin{align}
P_C(\Lambda)  &= \frac{1}{(\sqrt{2\pi}\sigma)^n}\int_{\vor(\Lambda)} 1
- \frac{\|\xbf\|^2}{2\sigma^2} + \frac{\left(\|\xbf\|^2\right)^2}{
4\sigma^42!} - \dots d\xbf \nonumber \\
&= \frac{1}{(\sqrt{2\pi}\sigma)^n} \sum_{m=0}^\infty
\frac{(-1)^m}{2^m\sigma^{2m}m!} \int_{\vor(\Lambda)} \|\xbf\|^{2m}
d\xbf.  \label{eq:summomentproberror}
\end{align}
So, to obtain arbitrarily accurate approximations to the probability of error it is enough to know the values of $\int_{\vor(\Lambda)} \|\xbf\|^{2m} d\xbf$ for $m=1,2\dots$ for some sufficiently large $m$.  The number of terms required will increase as the noise variance gets smaller.  This implies that the bound with a fixed number of terms is not asymptotically tight but is very accurate up to a finite signal-to-noise ratio.

\newcommand{\calM}{M}
In this paper we focus on $n$-dimensional lattice codes constructed from the family of lattices called $A_n$ and  we will derive expressions for the $m$th \emph{moment}s
\[
 \calM_n(m) = \int_{\vor(A_n)} \|\xbf\|^{2m} d\xbf.
\]
These can be summed in (\ref{eq:summomentproberror}) to give arbitrarily accurate approximations for the probability of error.

\section{The main result}\label{sec:main-result}

We now state our main result.  The $m$th moment $\calM_n(m)$ of the lattice $A_n$ satisfies
\begin{equation}\label{eq:theCmformula}
\frac{\calM_n(m)}{m!} = \frac{n\sqrt{n+1}}{n+2m}\sum_{k=0}^{m}\sum_{a =0}^{k}\sum_{b=0}^{k-a} \frac{G(n-1,a,2k - 2a - b)}{(-1)^{k-a}H(n,m,k,a,b)}
\end{equation}
where the function
\[
H(n,m,k,a,b) = \frac{(n+1)^{m-a}a!(m-k)!b! (k-a-b)!}{2^{b} n^{m-k}},
\]
and the function $G(n,c,d)$ satisfies the recursion
\begin{equation}\label{eq:theGrecursion}
G(n,c,d) = \sum_{c'=0}^{c} \sum_{d'=0}^{d} \binom{c}{c'}\binom{d}{d'} \frac{G(n-1,c-c',d-d')}{2c'+d'+1}
\end{equation}
with the initial conditions
\[
G(1,c,d) = \frac{1}{2c+d+1} \qquad \text{and} \qquad G(n,0,0) = 1.
\]
For fixed $m$ it is possible to solve this recursion in $n$ and obtain formula for the $\calM_n(m)$ in terms of $n$.  The first five such formula are:
\begin{align*}
\calM_n(0) &= \sqrt{n+1} \qquad \text{the volume of } \vor(A_n),\\
\calM_n(1) &= \frac{n(n+3)}{12\sqrt{n+1}} \qquad \text{the second moment of $\vor{A_n}$~ \cite[p. 462]{SPLAG}}, \\
\calM_n(2) &=  \frac{50 n + 55 n^2 + 34 n^3 + 5 n^4}{720 (1 + n)^{3/2}}, \\
\calM_n(3) &= \frac{1960 n + 2142 n^2 + 2681 n^3 + 1423 n^4 + 399 n^5 + 35 n^6}{60480 (1 + n)^{5/2}}, \\
\calM_n(4) &= \frac{93744 n + 34356 n^2 + 112172 n^3 + 89343 n^4 + 53224 n^5 + 17246 n^6 + 2940 n^7 + 175 n^8}{3628800 (1 + n)^{7/2}}.
\end{align*}
We have explicitly tabulated these formula for $m=0$ to $40$. For larger $m$ direct evaluation for specific $n$ from the recursive formula is preferable.  We will derive these results in Section~\ref{sec:powers-eucl-norm}, but first need some properties of the lattice $A_n$.

\section{The lattice $A_n$}\label{sec:lattice-a_n}
Let $H$ be the hyperplane orthogonal to the all ones vector of length $n+1$, denoted by $\onebf$, that is
\[
\onebf = \left[ \begin{array}{cccc} 1 & 1 & \cdots & 1 \end{array} \right]^\prime,
\]
where superscript $^\prime$ indicates the transpose.  Any vector in $H$ has the property that the sum (and therefore the mean) of its elements is zero and for this reason $H$ is often referred to as the \emph{zero-sum plane} or the \emph{zero-mean plane}.
The lattice $A_n$ is the intersection of the integer lattice $\ints^{n+1}$ with the zero-sum plane, that is
\begin{equation}
\label{eq:An_sub_Zn}
  A_{n} = \ints^{n+1} \cap H = \big\{ \xbf \in \ints^{n+1} \mid \dotprod{\xbf}{\onebf} = 0  \big\}.
\end{equation}
Equivalently, $A_n$ consists of all of those points in $\ints^{n+1}$ with coordinate sum equal to zero.
The lattice has $n(n+1)$ minimal vectors, each of squared Euclidean length $2$, so the packing radius is $\frac{1}{\sqrt{2}}$.  The $n$-volume of the Voronoi cell $\vor(A_n)$ is $\sqrt{n+1}$~\cite[p. 108]{SPLAG}.

The Voronoi cell of $A_n$ is closely related to the $n+1$ dimensional hypercube $\vor(\ints^{n+1})$ as the next theorem will show.  This result has appeared previously~\cite{McKilliam2009CoxeterLattices,McKilliam2010thesis}, but we repeat it here so that this paper is self contained.  We denote by
\[
\Qbf = \Ibf - \frac{\onebf\onebf^\prime}{\onebf^\prime \onebf} = \Ibf - \frac{\onebf\onebf^\prime}{n+1}
\]
the projection matrix orthogonal to $\onebf$ (i.e. into the zero-sum plane) where $\Ibf$ is the $n+1$ by $n+1$ identity matrix.  Given a set $S$ of vectors from $\reals^{n+1}$ we write $\Qbf S$  to denote the set with elements $\Qbf \sbf$ for all $\sbf \in S$, i.e. the set containing the projection of the vectors from $S$.

\begin{lemma} \label{lem:QVorZnsubsetVorAn}
The projection of $\vor(\ints^{n+1})$ into the zero-sum plane is a subset of $\vor(A_{n})$.  That is,
\[
\Qbf\vor(\ints^{n+1}) \subseteq \vor(A_{n}).
\]
\end{lemma}
\begin{IEEEproof}
Let $\ybf \in \vor(\ints^{n+1})$.  Decompose $\ybf$ into orthogonal components so that $\ybf = \Qbf \ybf + t \onebf$ for some $t \in \reals$.  Then $\Qbf\ybf \in \Qbf\vor(\ints^{n+1})$.  Assume that $\Qbf\ybf \notin \vor(A_n)$.  Then there exists some $\xbf \in A_n$ such that
\begin{align*}
\|\xbf - \Qbf\ybf\|^2 < \|\zerobf - \Qbf\ybf\|^2 & \Rightarrow \|\xbf - \ybf + t\onebf\|^2 < \|\ybf - t\onebf\|^2 \\
& \Rightarrow \|\xbf - \ybf\|^2 + 2t\xbf'\onebf < \|\ybf\|^2.
\end{align*}
By definition \eqref{eq:An_sub_Zn} $\xbf'\onebf = 0$ so $\|\xbf - \ybf\|^2 < \|\ybf\|^2$.  This violates that $\ybf \in \vor(\ints^{n+1})$ and hence $\Qbf\ybf \in \vor(A_n)$.\footnote{This proof can be generalised to show that for any lattice $L$ and hyperplane $P$ such that $P\cap L$ is also a lattice it is true that $p\vor(L) \subseteq \vor(L \cap P)$ where $p$ indicates the orthogonal projection into $P$~\cite[Lemma~2.1]{McKilliam2010thesis}.}
\end{IEEEproof}

\begin{theorem}  \label{thm:VorAn=QVorZn1}
The projection of $\vor(\ints^{n+1})$ into the zero-sum plane is equal to $\vor(A_{n})$. That is,
\[
\vor(A_n) = \Qbf\vor(\ints^{n+1}).
\]
\end{theorem}
\begin{proof}
Let $\ebf_i$ denote a vector with $i$th element equal to one and the remaining elements zero.  The $n$-volume of $\vor(A_n)$ is $\sqrt{n + 1}$.  From Berger~et.~al.~\cite[Theorem 1.1]{Burger1996} we find that the $n$-volume of the projected hypercube $\Qbf\vor(\ints^{n+1})$ is equal to
\[
\sum_{i = 1}^{n+1} \frac{\dotprod{\onebf}{\ebf_i}}{\|\onebf\|} =  \sum_{i = 1}^{n+1} \frac{1}{\sqrt{n+1}} = \sqrt{n + 1}
\]
also. It follows from Lemma~\ref{lem:QVorZnsubsetVorAn} that $\Qbf\vor(\ints^{n+1}) \subseteq \vor(A_n)$, so, because the volumes are the same, and because $\vor(A_n)$ and $\Qbf\vor(\ints^{n+1})$ are polytopes, we have $\vor(A_n) = \Qbf\vor(\ints^{n+1})$.
\end{proof}

\section{Integrating a function over $\vor(A_n)$}\label{sec:integr-funct-over}

We would like to be able to integrate functions over the Voronoi cell of $A_n$.  Consider a function $f : \reals^{n+1} \mapsto \reals$.  The definition we have made for $A_n$ above places it in the $n$-dimensional zero-sum plane, lying in $\reals^{n+1}$.  The Voronoi cell is a subset of the zero-sum plane that has zero $n+1$-dimensional volume.  So, the integral $\int_{\vor(A_n)} f(\xbf) d \xbf$ is equal to zero.  This is not what we intend.  By an appropriate change of variables it would be possible to write the Voronoi cell $\vor(A_n)$ in an $n$-dimensional coordinate system, and then integrate.  However, we find the following approach simpler.  Given a set $S$ of vectors from the zero sum plane, let $S \times \onebf$ denote the set of elements that can be written as $\xbf + \ybf$ where $\xbf \in S$ and $\ybf = k\onebf$ for some $k \in [-\nicefrac{1}{2},\nicefrac{1}{2}]$.  Now, the integral over the Voronoi cell can be written as
\begin{equation}\label{eq:vorfint}
\int_{\vor(A_n) \times \onebf} f(\Qbf\xbf) d\xbf = \int_{\Qbf\vor(\ints^{n+1}) \times \onebf} f(\Qbf\xbf) d\xbf.
\end{equation}
It is not immediately clear how an integral over $\vor(A_n) \times \onebf$ should be performed.  Consider the following simpler integral over the hypercube $\vor(\ints^{n+1})$,
\begin{equation}\label{eq:unormalised}
\int_{\vor(\ints^{n+1})} f(\Qbf\xbf) d\xbf.
\end{equation}
This integral is not equal to~\eqref{eq:vorfint} because, although $\Qbf\xbf$ is always an element of $\vor(A_n)$, the integral is not uniform over $\vor(A_n)$.  To see this, consider some $\xbf \in \vor(\ints^{n+1})$ and let $x_{\text{max}}$ be the maximum element of $\xbf$ and $x_{\text{min}}$ be the minimum element.  Then $\xbf +  k \onebf \in \vor(\ints^{n+1})$ for those $k \in [-\nicefrac{1}{2} - x_{\text{min}}, \nicefrac{1}{2} - x_{\text{max}})$.  The length of this interval is $1 - x_{\text{max}} + x_{\text{min}}$ so the (one dimensional) volume of points in $\vor(\ints^{n+1})$ that, once projected orthogonally to $\onebf$, are equal to $\Qbf\xbf$ is given by
\[
\|\onebf\|(1 - x_{\text{max}} + x_{\text{min}}) = \sqrt{n+1}(1 - x_{\text{max}} + x_{\text{min}}).
\]
The integral~\eqref{eq:vorfint} can be obtained by normalising~\eqref{eq:unormalised} by this length, that is,
\[
\int_{\vor(A_n) \times \onebf} f(\Qbf\xbf) d\xbf =  \int_{\vor(\ints^{n+1})} \frac{f(\Qbf\xbf)}{\sqrt{n+1}(1 - x_{\text{max}} + x_{\text{min}})}  d\xbf.
\]
The primary advantage of this integral is that the bounds are given by the hypercube $\vor(\ints^{n+1})$.

Let us now restrict $f(\xbf)$ so that it depends only on the magnitude $\|\xbf\|$, for example $f(\xbf) = \|\xbf\|^{2m}$ could be a power of the Euclidean norm of $\xbf$.  Now $f(\xbf)$ is invariant to permutation of $\xbf$.  Let $\xbf$ be such that $x_1$ is the maximum element and $x_2$ is the minimum element.  Our integral is now equal to
\[
\frac{n (n+1)}{\sqrt{n+1}} \int^{1/2}_{-1/2} \int^{x_1}_{-1/2}  \int^{x_1}_{x_2} \cdots \int^{x_1}_{x_2} \frac{f(\Qbf\xbf)}{1 - x_{1} + x_{2}}  dx_{n+1} \, \dots \, dx_2 \, dx_1.
\]
The factor $n(n+1)$ arises because there are $n(n+1)$ ways to place two elements (i.e. $x_1$ and $x_2$) into $n+1$ positions.

We can make further simplifications.  Letting $t = x_1 - x_2$ and $y = x_1 + 1/2$ and changing variables gives
\[
n \sqrt{n+1} \int^{1}_{0} \int^{y}_{0}  \int^{y - 1/2}_{y - t - 1/2} \cdots \int^{y-1/2}_{y-t-1/2} \frac{f(\Qbf\xbf )}{1 - t}  dx_{n+1} \, \dots \, dx_3 \, dt \, dy,
\]
and letting $w_{i-2} = x_i - y + 1/2 + t$ for $i = 3,\dots,n+1$ gives
\begin{equation}\label{eq:intwewilluse}
n \sqrt{n+1} \int^{1}_{0} \int^{y}_{0}  \int^{t}_{0} \cdots \int^{t}_{0} \frac{f(\Qbf\xbf )}{1 - t}  dw_{n-1} \, \dots \, dw_1 \, dt \, dy.
\end{equation}
Observe that $\xbf = \wbf + (y - t - \nicefrac{1}{2})\onebf$ where $\wbf$ is the column vector
\[
\wbf = [t, 0, w_1, w_2, \dots, w_{n-1}]^\prime.
\]
Projecting orthogonal to $\onebf$ gives $\Qbf\xbf = \Qbf \wbf$.  Interestingly $\wbf$ does not contain $y$ so the term inside the integral does not depend on $y$.  This is the integral we will use to compute the moments of $A_n$.

\begin{example}\emph{\textbf{(The volume of the Voronoi cell)}
In order to demonstrate this approach we will derive the $0$th moment (i.e. the volume) of the Voronoi cell using~(\ref{eq:intwewilluse}).  Setting $f(\Qbf\wbf) = \|\Qbf\wbf\|^0 = 1$ we obtain,
\begin{align*}
M_n(0) &= n\sqrt{n+1} \int^{1}_{0} \int^{y}_{0} \int^{t}_{0} \cdots \int^{t}_{0} \frac{1}{1 - t} \,dw_{n-1}\dots dw_1 \, dt  \, dy\\
 &= n\sqrt{n+1} \int^{y}_{0} \int_{0}^{1}  \frac{t^{n-1}}{1 - t} \, dt  \, dy. \\
&= n\sqrt{n+1}  \int^{1}_{0} \beta(y, n, 0) \, dy = \sqrt{n+1}.
\end{align*}
as required.  Here $\beta(x,a,b) =  \int_{0}^{x}  t^{a-1}(1 - t)^{b-1} \,dt$ is the incomplete beta function~\cite{Pearson_tables_of_beta_functions} and we have used the identity $\int^{1}_{0} \beta(y, n, 0) \, dy = \frac{1}{n}$.}
\end{example}

\section{The moments of $A_n$}\label{sec:powers-eucl-norm}

We now derive expressions for the $\calM_n(m)$.  Setting $f(\Qbf\xbf) = \left(\|\Qbf\wbf\|^{2}\right)^m$ in~\eqref{eq:intwewilluse} we obtain,
\[
\frac{\calM_n(m)}{n\sqrt{n+1}} = \int_0^1 \int_0^y \int_0^t \cdots \int_0^t \frac{\left(\|\Qbf\wbf\|^2\right)^m}{1 - t} \,dw_{n-1} \, \dots \, dw_1 \, dt \, dy.
\]
Now $\|\Qbf\wbf\|^2 = \|\wbf\|^2 - \frac{1}{n+1}(\wbf^\prime \onebf)^2$ and recalling that $\wbf = [t,0,w_1,\dots,w_{n-1}]'$ we can write
\begin{align*}
\|\Qbf\wbf\|^2 &= \|\wbf\|^2 - \frac{1}{n+1}(\wbf^\prime \onebf)^2 \\
&= t^2 + \sum_{i=1}^{n+1}w_i^2 - \frac{1}{n+1}\left( t + \sum_{i=1}^{n+1}w_i \right)^2 \\
&= t^2 + \sum_{i=1}^{n+1}w_i^2 - \frac{1}{n+1}\left( t^2 + 2 t \sum_{i=1}^{n+1}w_i + \left(\sum_{i=1}^{n+1}w_i \right)^2 \right)\\
&= C + D,
\end{align*}
say, where
\[
C = \left(\frac{n}{n+1} \right) t^2 \qquad \text{and} \qquad D = A - \frac{2t}{n+1}B - \frac{1}{n+1}B^2,
\]
and where,
\[
 A = \sum_{i=1}^{n-1}w_i^2 \qquad \text{and} \qquad B = \sum_{i=1}^{n-1}w_i.
\]
Now,
\[
 \frac{\calM_n(m)}{n\sqrt{n+1}} =  \int_0^1 \int_0^y \int_0^t \cdots \int_0^t \frac{(C+D)^m}{1 - t} \,dw_{n-1} \, \dots \, dw_1 \, dt \, dy,
\]
and by expanding the binomial $(C+D)^m$ we get
\[
\frac{\calM_n(m)}{n\sqrt{n+1}} =  \int_0^1 \int_0^y \frac{1}{1 - t}\sum_{k=0}^{m}\binom{m}{k} C^{m-k} \int_0^t \cdots \int_0^t \, D^{k} \,dw_{n-1} \, \dots \, dw_1 \, dt \, dy.
\]
Expanding $D^k$ as a trinomial gives
\begin{align*}
D^k &= \sum_{k_1+k_2+k_3=k} \frac{k! A^{k_1} B^{2k_3 + k_2}}{k_1! k_2! k_3!} \left(\frac{-1}{n+1}\right)^{k_2+k_3}2^{k_2}t^{k_2}  \\
&=  \sum_{a=0}^k\sum_{b=0}^{k-a} \frac{k!  A^{a} B^{2k - 2a - b}}{a! b! (k-a-b)!} \left(\frac{-1}{n+1}\right)^{k - a}2^{b}t^{b}
\end{align*}
where the second line follows by setting $k_1 = a$, $k_2 = b$ and $k_3 = k - a - b$.  In Appendix~\ref{sec:mult-type-integr} we show that the integral of $A^{a} B^{2k - 2a - b}$ over $w_1,\dots w_{n-1}$ is
\begin{equation} \label{eq:recurseint}
\int_0^t \cdots \int_0^t A^{a} B^{2k-2a-b} dw_{n-1} \, \dots \, dw_1 = t^{n-1+2k - b} G(n-1,a,2k - 2a - b).
\end{equation}
where $G(n,c,d)$ satisfies the recursion given by~\eqref{eq:theGrecursion}.  So, let $P$ satisfy
\begin{align*}
P &= t^{1 - n - 2k}\int_0^t \cdots \int_0^t D^k dw_{n-1} \, \dots \, dw_1 \\
&= \sum_{a=0}^k\sum_{b=0}^{k-a} \frac{2^{b} k!G(n-1,a,2k - 2a - b)}{a! b! (k-a-b)!} \left(\frac{-1}{n+1}\right)^{k - a}.
\end{align*}
Now $C^{m-k} = \left(\frac{n}{n+1} \right)^{m-k} t^{2(m-k)}$ and
\begin{align*}
\frac{\calM_n(m)}{n\sqrt{n+1}} &= \sum_{k=0}^{m} \binom{m}{k} \left(\frac{n}{n+1} \right)^{m-k} P \int_0^1 \int_0^y \frac{t^{n - 1 + 2m}}{1 - t} \, dt \, dy \\
&= \sum_{k=0}^{m} \binom{m}{k} \left(\frac{n}{n+1} \right)^{m-k} P \int_0^1 \beta(y,n+2m,0) dy \\
&= \frac{1}{n+2m} \sum_{k=0}^{m} \binom{m}{k} \left(\frac{n}{n+1} \right)^{m-k} P.
\end{align*}
This expression is equivalent to that from~(\ref{eq:theCmformula}).

\section{Results and simulations}\label{sec:results-simulations}

We now plot the probability of coding error versus signal to noise ratio (SNR) for the lattices $A_1, A_2, A_3, A_4, A_5$ and $A_8$.  For these plots the SNR is related to noise variance according to~\cite{Viterbo_diamond_cutting_1996}
\[
\text{SNR} = \frac{V^{2/n}}{4\sigma^2},
\]
where $V$ is the volume of the Voronoi cell and $n$ is the dimension of the lattice.  Figure~\ref{fig:peplots} shows the `exact' probability of error (correct to 16 decimal places) computed using the moments $\calM_n(m)$ and~\eqref{eq:summomentproberror} (solid line).  The number of moments needed to ensure a certain number of decimal places accuracy depends on $n$ and also on the noise variance $\sigma^2$.  At most 321 moments where needed for Figure~\ref{fig:peplots}.  We also display the probability of error computed approximately by Monte-Carlo simulation (dots).  
The simulations are iterated until 500 error events occur.

The plot also display an approximation for the probability of error for the 8-dimensional $E_8$ lattice.  The approximation is made in the usual way by applying the union bound to the minimal vectors of the lattice~\cite[p.~71]{SPLAG}.  The $E_8$ lattice has 240 minimal vectors of length $\sqrt{2}$.  The packing radius of $E_8$ is therefore $\rho = \sqrt{2}/2$.  Applying the union bound the probability of error satisfies
\newcommand{\erfc}{\operatorname{erfc}}
\newcommand{\erf}{\operatorname{erf}}
\[
P_E \leq 240\erfc\left( \frac{\rho}{\sqrt{2}\sigma} \right) = 240\erfc\left(\frac{1}{2\sigma}\right)
\]
where $\erfc(x) = 1 - \erf(x)$ is the complementary error function.  For the $E_8$ lattice this approximation is an upper bound because the relevant vectors of $E_8$ (those vectors that define the Voronoi cell) are precisely the 240 minimal vectors.  

\begin{figure*}[tp]
	\centering
		\includegraphics{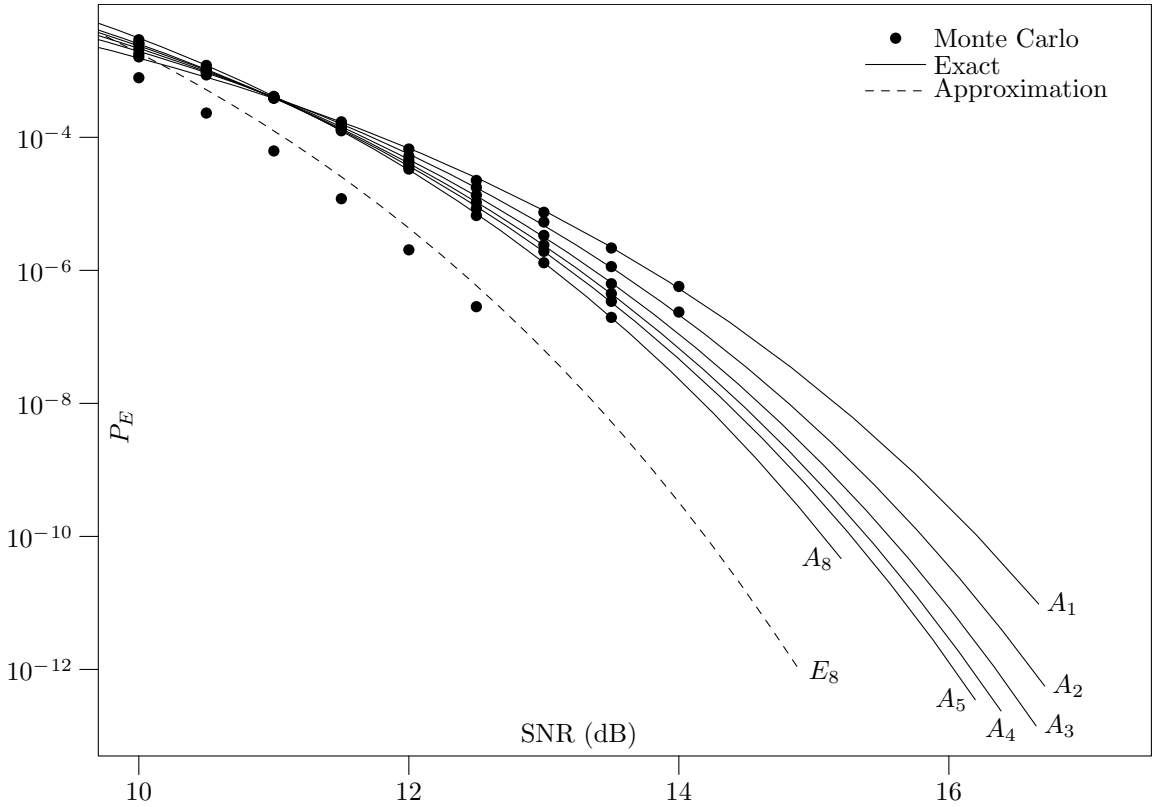}
		\caption{The probability of error versus SNR for $A_1 \simeq \ints$, $A_2$, $A_3\simeq D_3, A_4, A_5, A_8$ and $E_8$.}
		\label{fig:peplots}
\end{figure*}

\section{Conclusion}

Recursive formulae for the moments of the Voronoi cell of the lattice $A_n$ were found.  These enable accurate prediction of the performance of codes constructed from $A_n$.  In two dimensions $A_2$ is the hexagonal lattice and in three dimensions $A_3$ is the face-centered cubic lattice.  Our results include codes constructed using these packings as a special case.  

\appendix

\subsection{A multinomial type integral over a hypercube}\label{sec:mult-type-integr}

In~(\ref{eq:recurseint}) we required to evaluate integrals of the form
\[
F(n-1,a,2k - 2a - b) = \int^{t}_{0} \cdots \int^{t}_{0} A^a B^{2k - 2a - b} \,dw_{n-1}\cdots dw_{1},
\]
or equivalently, integrals of the form
\[
F(n,c,d) = \int^{t}_{0} \cdots \int^{t}_{0} \left(\sum_{j=1}^{n} x_j^2\right)^c \left( \sum_{i=1}^{n} x_i \right)^d \,dx_1\cdots dx_{n}
\]
where $n,c$ and $d$ are integers.  We shall find a recursion describing this integral.  Write
\[
F(n,c,d) = \int^t_0 \cdots \int^t_0 \left(x_n^2 + \sum_{j=1}^{n-1} x_j^2\right)^c \left( x_n + \sum_{i=1}^{n-1} x_i \right)^d \,dx_1\cdots dx_n.
\]
Expanding the two binomials gives
\begin{align*}
F(n,c,d) &=  \int^t_0 \cdots \int^t_0 \sum_{c'=0}^{c} \binom{c}{c'} x_n^{2c'} \left(\sum_{j=1}^{n-1} x_j^2\right)^{c-c'}  \sum_{d'=0}^{d}\binom{d}{d'} x_n^{d'} \left(\sum_{i=1}^{n-1} x_i \right)^{d-d'} \,dx_1\cdots dx_{n} \\
&=  \sum_{c'=0}^{c} \sum_{d'=0}^d \int^t_0 \cdots \int^t_0 \binom{c}{c'} \binom{d}{d'} x_n^{2c'+d'} \left(\sum_{j=1}^{n-1} x_j^2\right)^{c-c'} \left(\sum_{i=1}^{n-1} x_i \right)^{d-d'} \,dx_1\cdots dx_{n}.
\end{align*}
Integrating the $x_n$ term gives
\[
F(n,c,d) = \sum_{c'=0}^{c} \sum_{d'=0}^{d} \binom{c}{c'}\binom{d}{d'} \frac{t^{2c'+d'+1}}{2c'+d'+1} \int^t_0 \cdots \int^t_0 \left(\sum_{j=1}^{n-1}x_j^2\right)^{c-c'}\left(\sum_{i=1}^{n-1} x_i \right)^{d-d'} \,dx_1\cdots dx_{n-1}.
\]
Note that
\[
F(n-1,c-c',d-d')  = \int^t_0 \cdots \int^t_0 \left(\sum_{j=1}^{n-1}x_j^2\right)^{c-c'}\left(\sum_{i=1}^{n-1} x_i \right)^{d-d'} \,dx_1\cdots dx_{n-1}.
\]
So $F(n,c,d)$ satisfies the recursion
\[
F(n,c,d) = \sum_{c'=0}^{c} \sum_{d'=0}^{d} \binom{c}{c'}\binom{d}{d'} \frac{t^{2c'+d'}}{2c'+d'+1}F(n-1,c-c',d-d')
\]
with the initial conditions
\[
F(1,c,d) = \frac{t^{2c+d+1}}{2c+d+1} \qquad \text{and} \qquad F(n,0,0) = t^n.
\]
The $F(n,c,d)$ can be written as $t^{n+2c+d}G(n,c,d)$ where $G(n,c,d)$ is rational.  To see this write
\begin{align*}
F(n,c,d) &=  \sum_{c'=0}^{c} \sum_{d'=0}^{d} \binom{c}{c'}\binom{d}{d'} \frac{t^{2c'+d'+1}}{2c'+d'+1}t^{n-1+2(c-c')+d-d'}G(n-1,c-c',d-d') \\
&=  t^{n+2c+d}\sum_{c'=0}^{c} \sum_{d'=0}^{d} \binom{c}{c'}\binom{d}{d'} \frac{G(n-1,c-c',d-d')}{2c'+d'+1} \\
&=  t^{n+2c+d}G(n,c,d).
\end{align*}
Now $G(n,c,d)$ is the rational number satisfying the recursion
\[
G(n,c,d) = \sum_{c'=0}^{c} \sum_{d'=0}^{d} \binom{c}{c'}\binom{d}{d'} \frac{G(n-1,c-c',d-d')}{2c'+d'+1}
\]
with the initial conditions
\[
G(1,c,d) = \frac{1}{2c+d+1} \qquad \text{and} \qquad G(n,0,0) = 1.
\]

\subsection{Solving this recursion for fixed $d$ and $c$}\label{sec:solv-this-recurs}
\newcommand{\calG}{\mathcal G}

For fixed $d$ and $c$ this recursion can be solved explicitly.  Write
\[
G(n,c,d) = G(n-1,c,d) + \sum_{ (c',d') \neq (0,0)} \binom{c}{c'}\binom{d}{d'} \frac{G(n-1,c-c',d-d')}{2c'+d'+1}
\]
where the sum $\sum_{ (c',d') \neq (0,0)}$ is over all $0 \leq c' \leq c$ and $0 \leq d' \leq d$ except when both $d$ and $c$ are zero.  Denote by $\calG(z,c,d)$ the $z$-transform of $G(n,c,d)$.  Taking the $z$-transform of both sides in the equation above gives
\[
\calG(z,c,d) = \frac{z^{-1}}{1-z^{-1}} \sum_{ (c',d') \neq (0,0)} \binom{c}{c'}\binom{d}{d'} \frac{\calG(z,c-c',d-d')}{2c'+d'+1}.
\]
So the $z$-transform $\calG(z,c,d)$ satisfies this recursive equation.  The initial condition is $\calG(z,0,0) = \frac{z^{-1}}{1 - z^{-1}}$.  By inverting this $z$-transform and using the resultant expressions in~\eqref{eq:theCmformula} we obtain formulae in $n$ for the moment $\calM_n(m)$.  This procedure was used to generate the formula described in Section~\ref{sec:main-result}.  Mathematica 8.0 was used to perform these calculations.


\small
\bibliography{bib}






\end{document}